\documentclass[12pt]{iopart}
\usepackage{iopams}
\usepackage{graphicx}
\usepackage{amsthm}

\newcommand{\Li}{\textnormal{Li}}
\newcommand{\var}{\textnormal{var}}
\newtheorem{theorem}{Theorem}[section]
\newtheorem{corollary}[theorem]{Corollary}
\newtheorem{lemma}[theorem]{Lemma}

\makeatletter
\newcommand{\doublewidetilde}[1]{{%
  \mathpalette\double@widetilde{#1}%
}}
\newcommand{\double@widetilde}[2]{%
  \sbox\z@{$\m@th#1\widetilde{#2}$}%
  \ht\z@=.85\ht\z@
  \widetilde{\box\z@}%
}
\makeatother

\begin{document}

\title[Consensus time in a voter model with concealed and publicly
expressed opinions]{Consensus time in a voter model with concealed and publicly expressed opinions}

\author{Michael T. Gastner$^1$, Be\'ata Oborny$^2$ and M\'at\'e Guly\'as$^3$}

\address{$^1$ Yale-NUS College, Division of Science, 16 College Avenue
  West, \#01-220 Singapore 138527\\
$^2$ Department of Plant Taxonomy, Ecology and Theoretical Biology,
Biological Institute, Lor\'and E\"otv\"os University (ELTE), P\'azm\'any
P. stny. 1/c, Budapest, H-1117, Hungary\\
$^3$ ``Lend\"ulet'' Research Center for Educational and Network Studies, Centre for Social Sciences, Hungarian Academy of Sciences, T\'oth K\'alm\'an street 4, Budapest, H-1097, Hungary}
\ead{\mailto{michael.gastner@yale-nus.edu.sg},
  \mailto{beata.oborny@ttk.elte.hu}, \mailto{mategulyas@gmail.com}}
\begin{abstract}
The voter model is a simple agent-based model to mimic opinion dynamics in social
networks: a randomly chosen agent adopts the opinion of a randomly chosen neighbour. 
This process is repeated until a consensus emerges. 
Although the basic voter model is theoretically intriguing, it misses an important feature of real opinion dynamics: it does not distinguish between an agent's publicly expressed opinion and her inner conviction. 
A person may not feel comfortable declaring her conviction if her social circle appears to hold an opposing view. 
Here we introduce the Concealed Voter Model where we add a second, concealed layer of opinions to the public layer.
If an agent's public and concealed opinions disagree, she can reconcile them by either publicly disclosing her previously secret point of view or by accepting her public opinion as inner conviction. 
We study a complete graph of agents who can choose from two opinions. We define a martingale $M$ that determines the probability of all agents eventually agreeing on a particular opinion. 
By analyzing the evolution of $M$ in the limit of a large number of agents, we derive the leading-order terms for the mean and standard deviation of the consensus time (i.e.\ the time needed until all opinions are identical). We thereby give a precise prediction by how much concealed opinions slow down a consensus.

\end{abstract}

\maketitle

\section{Introduction}
\label{sec:intro}

The voter model, introduced in the 1970s~\cite{CliffordSudbury73,HolleyLiggett75}, has became a paradigmatic model for the theoretical study of opinion dynamics~\cite{Liggett99,CastellanoEtAl09}. 
It describes the emergence of an ordered state (e.g.\ a consensus) through the local interactions of individual agents (e.g.\ voters).  
This phenomenon occurs in various contexts. Therefore, the original voter model and its extended versions have been broadly applied not only in social and political sciences~\cite{CastellanoEtAl09,FernandezGracia_etal14}, but also in chemistry (e.g.\ in the study of catalytic  reactions~\cite{FrachebourgKrapivsky96}) and biology (e.g.\ for modelling ecological competition~\cite{DurrettLevin96,ChaveLeigh02,BorileEtAl14} and prey-predator interaction~\cite{RavaszEtAl04}). 

Most voter models share the features that
\begin{enumerate}
\item
each agent is in one state (e.g.\ has one particular opinion) out of two alternatives and
\item
the state can only be changed through pairwise interactions between agents.
\end{enumerate}
(For exceptions from rules (i) and (ii), see~\cite{VazquezConstrained04} and~\cite{LambiotteVacillating07}, respectively.)
In each interaction, a focal agent is selected at random together with a randomly chosen adjacent agent. 
The graph that describes the adjacency can be complete, regular (e.g.\ a square lattice) or more complex (see examples in~\cite{Sood_etal08} and~\cite{MasudaOpControl15}). 
The focal agent can adopt or reject the opinion of the adjacent agent. 
In the original version of the voter model, which we call the {\it Basic Voter Model} (BVM), the rule of interaction is simple: the focal agent must adopt the neighbour's opinion.

The BVM on a finite, strongly connected graph (i.e.\ a graph in which there is a directed path between every pair of agents) inevitably reaches a global consensus~\cite{Serrano_etal09} (i.e.\ an absorbing state in which every agent holds the same opinion).  
Various versions of the model have applied more complex rules~\cite{CastellanoEtAl09}, some of which may {\it not} guarantee that all agents ultimately share the same opinion. 
However, in this paper we only consider models that {\it must} end in a consensus. 
In this context, an important question is how long it takes until all agents agree. The consensus time (also called exit time, hitting time or first-passage time) depends on the rules of interaction and graph structure~\cite{SoodRedner05}.

\begin{figure}
\begin{center}
\includegraphics[width=0.36\textwidth]{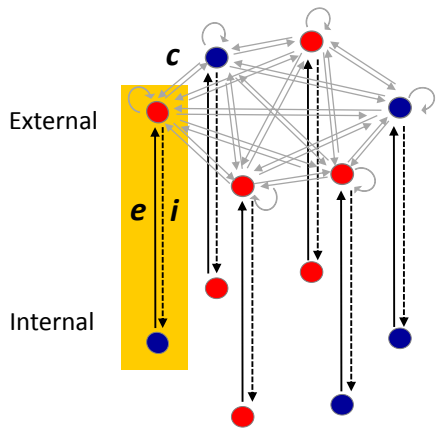}
\caption{\label{fig:modeldescr}Example for the CVM in a small, complete network, with $N=6$.  Each agent is represented by two nodes, one in the external and one in the internal layer. The state of the agent marked by a yellow background is $Rb$; note that we use upper-case letters for the external and lower-case letters for the internal opinions. The following elementary events can happen to this agent: copying the external opinion of a randomly selected agent (solid gray arrow); externalization (solid black arrow); or internalization (dotted black arrow). The corresponding rates are $c$, $e$ and $i$. The loops in the external layer indicate that copying the agent's own opinion is also among the options; in that case, the agent keeps the original opinion.}
\end{center}
\end{figure}

In this article, we introduce a new kind of model, the {\it Concealed Voter Model} (CVM), which differs from the BVM by distinguishing between an agent's publicly expressed opinion and her inner conviction about the particular subject.
For this reason, we add a second, concealed layer of opinions to the public layer (Figure~\ref{fig:modeldescr}).
The duality between inner conviction and publicly expressed opinion is an important phenomenon in every society (see examples in~\cite{Kuran09}). 
To our knowledge, the first theoretical model that explicitly considered this duality was the ``partisan'' voter model~\cite{MasudaHeteroVot10,MasudaRednerTruth10}.  
Every agent in that model is a ``partisan'' with a fixed and innate preference for one of the two opinions, which influences the agent's publicly expressed opinion. 
Thus, the partisan model distinguishes between two layers, but the interaction between the layers is only unidirectional: the internal opinion influences the external one, but the partisan voter model does not contain any feedback mechanism from the external to the internal opinion.
The main novelty in the CVM is that it permits a bidirectional interaction between both layers.

We compare the mean and the standard deviation of the consensus time in the BVM and CVM. 
Thus, we investigate whether the existence of concealed opinions increases the coexistence time of alternative opinions in a group of people and, if yes, to what extent.
In real societies, concealed opinions are ubiquitous on many kinds of issues~\cite{Kuran09}, for example on political votes, debated social norms or consumption habits. 
Understanding the composition of opinions in the hidden layer, and how it influences the public, is important for making reliable predictions about collective opinion formation.

\section{The Basic and the Concealed Voter Model}
Let $N$ denote the number of agents in the system. 
In the BVM, the state of each agent $\alpha$ at time $t$ is defined as her
publicly expressed opinion $\omega_{\rm{ext}}(\alpha, t)$.
The subscript ``ext'' emphasizes that this is an {\it external} opinion in
the sense that any adjacent agent can learn that $\alpha$'s opinion is
$\omega_{\rm{ext}}(\alpha, t)$.
In our model, $t$ is continuous and opinions are updated asynchronously (see~\cite{Gastner15} for a discussion of differences between voter models with synchronous and asynchronous updates).
The opinion is binary.
We denote the options as red ($R$) or blue ($B$). 
The agent updates her opinion by selecting one of the adjacent agents,
say $\beta$, randomly and copying $\beta$'s opinion $\omega_{\rm{ext}}(\beta, t)$. 
The time intervals between two consecutive copying events by agent
$\alpha$ are exponentially distributed with a rate $c$. 
All time intervals are independent
of each other so that the times between two successive copying events in a system composed of $N$ agents is exponentially distributed with a rate $c{\cdot}N$.

In the CVM, the state of each agent is given by a pair of states $(
\omega_{\rm{ext}}(\alpha, t),\omega_{\rm{int}}(\alpha, t) )$, where $\omega_{\rm{ext}}$ is the {\it external} and $\omega_{\rm{int}}$ the {\it internal} opinion (Figure~\ref{fig:modeldescr}). 
We denote the two possible external opinions with a capital letter $R$ or $B$, whereas the lower-case letter $r$ or $b$ stands for her internal opinion. The agents' four possible states are denoted as $Rr$, $Rb$, $Br$ and $Bb$.
The rules of updating in the external layer are the same as in the BVM. 
An additional process, that is specific to the CVM, is introspection. When $\omega_{\rm{ext}} (\alpha, t)\neq \omega_{\rm{int}}(\alpha, t) $ (i.e.\ there is a discordance between the external and the internal opinion), the agent can relax it by either of two processes: {\it externalization}, whereby the formerly concealed internal state becomes public; or {\it internalization}, when the publicly voiced opinion becomes an internal conviction. Altogether, there are three competing processes in the system: copying, externalization, and internalization. The time intervals between two consecutive events of each kind are independent and exponentially distributed with rates $c{\cdot}N$, $e{\cdot}N$ and $i{\cdot}N$, respectively, in the group of $N$ agents.
Copying other agents' opinions promotes consensus within the external layer whereas externalization and internalization advance consensus between the two layers.

In this article, we study the global outcome of these processes for large $N$ in complete graphs. 
Because complete graphs are strongly connected, the agents must reach a consensus if $N$ is finite.
We include the reflexive relation (i.e.\ there is a loop from an agent to herself) for mathematical convenience.
For large $N$, the difference between the dynamics with and without the self-loop is negligible. In the BVM, we describe the state of the system at any given time $t$ by the fraction $\rho_R(t)$ of agents whose opinion is red in the only (i.e.\ external) layer. In the CVM, we describe the state by (i)  $\rho_R$ in the {\it external} layer, (ii) the fraction $\rho_r$ of red opinions in the {\it internal} layer, and (iii) the fraction $\rho_{Rr}$ of agents whose opinion is red in {\it both} layers. 
We study the time evolution of these observables and examine the consensus time (i.e.\ the first time when only a single opinion is present in {\it both} layers).

In section~\ref{sec:bvm}, we review some published results about the BVM concerning the main statistical properties of the consensus time.  In order to investigate the same characteristics in the CVM, we first introduce an intermediate model in section~\ref{sec:tlvm} (the Two-Layered Voter Model), in which we  also assume the existence of two connected layers. However,  unlike in the CVM, the event $\omega_{\rm{ext}}(\alpha, t)=R$ is independent of $\omega_{\rm{int}}(\alpha, t) = r$ so that the equations are easier to solve.
Finally, we show in section~\ref{sec:cvm} that the CVM equations can be solved similarly.

\section{The Basic Voter Model (BVM) on the complete graph}
\label{sec:bvm}
The mean consensus time of the BVM on the complete graph has been the subject of earlier publications (e.g.\ \cite{Sood_etal08,Redner01}).
Here we briefly review these results and also state the equation for higher moments of the consensus time distribution.
We include this review because the BVM acts as a base case for comparison with the CVM and motivates the analytic techniques we apply below.

We describe the state $S$ of the system by the fraction
$\rho_R\in\{0, \frac1N, \ldots, \frac{N-1}N, 1\}$ of agents whose
opinion is red.
The probability that the next copy increases $\rho_R$ by $\frac1N$
equals the product of the probability $(1-\rho_R)$ that the copying
agent's opinion is blue and the probability $\rho_R$ that the copied
agent holds the red opinion.
Because copying events happen with rate $cN$, the transition rate from
$\rho_R$ to $\rho_R + \frac1N$ is
\begin{equation}
  Q_{\rm{BVM}}\left(\rho_R, \rho_R + \frac1N\right) =
  cN\rho_R(1-\rho_R).
  \label{eq:QincBVM}
\end{equation}
Similarly, the rate with which $\rho_R$ decreases to $\rho_R-\frac1N$
is the product of the copy rate $cN$, the probability $\rho_R$ that
the copying agent is red and the probability $(1-\rho_R)$ that the
copied opinion is blue,
\begin{equation}
  Q_{\rm{BVM}}\left(\rho_R, \rho_R - \frac1N\right) = cN\rho_R(1-\rho_R).
  \label{eq:QdecBVM}
\end{equation}
A comparison between \eref{eq:QincBVM} and \eref{eq:QdecBVM} reveals that
both transition rates are in fact equal.

The probability that in one copy we increase or decrease $\rho_R$ by
more than $\frac1N$ is zero.
The diagonal element $Q(\rho_R, \rho_R)$ of a transition rate matrix
$\mathbf{Q}$ is conventionally defined as $-\sum_{x\neq\rho_R}
Q(\rho_R, x)$.
Therefore, the only nonzero matrix elements in the row corresponding to
state $\rho_R$ are $Q_{\rm{BVM}}\left(\rho_R,\rho_R \pm
  \frac1N\right)$ and
\begin{equation}
  Q_{\rm{BVM}}(\rho_R,\rho_R) = -2cN\rho_R(1-\rho_R).
  \label{eq:Q0BVM}
\end{equation}

The relatively simple structure of $\mathbf{Q}_{\rm{BVM}}$ allows us
to verify that the state $S(t)$ at time $t$ is a martingale,
\begin{equation}
  E[S(t+u) \mid S(t)] = S(t),
  \label{eq:BVMmart}
\end{equation}
where $E[\ldots | \ldots]$ denotes the conditional expectation value
and $u$ is an arbitrary nonnegative time.
Equation~\eref{eq:BVMmart} follows directly from the fact that
$\sum_x x Q_{\rm{BVM}}(\rho_R,x) = 0$ for all $\rho_R$.
If we start a series of Monte Carlo simulations always from a fixed
initial fraction $S(0)=s_0$ of red agents, the martingale property
\eref{eq:BVMmart} manifests itself as follows.
As we increase the number of simulations, the sample mean of $S(t)$ --
averaged over different simulations, but at a fixed $t$ -- converges
to $s_0$ although each individual run is likely to differ from $s_0$~\cite{Suchecki_etal05}.

We denote the consensus time by $T_{\rm{cons}}$.
For every realization, we can tell whether $T_{\rm{cons}}=t$ is true without having to know any of the states at times $>t$.
Therefore, in the parlance of stochastic processes, $T_{\rm{cons}}$ is a ``stopping time''.
Because we also know that $S$ is a martingale with
time-independent lower and upper bounds $0$ and $1$, we must have
$E[S(T_{\rm{cons}})] = E[S(0)]$ (see~\cite{Durrett16}).
Moreover, the only two possible states at $T_{\rm{cons}}$ are either
$S(T_{\rm{cons}}) = 1$ (i.e.\ a red consensus) or $S(T_{\rm{cons}})=0$
(i.e.\ a blue consensus) so that the probability of a
red consensus is given by the expected initial fraction $E[S(0)]$ of red
agents.
In the special case where all simulations start with $S(0)=s_0$, we
reach a red consensus with probability $s_0$.

How long does it take on average to reach a consensus from $s_0$?
If we denote the expected consensus time by
\[
\mu^{(1)}_{\rm{BVM}}(s_0) = E[T_{\rm{cons}} \mid S(0)=s_0],
\]
where the superscript ``(1)'' indicates that the left-hand side is the {\it first} moment of the conditional consensus time distribution,
then $\mu_{\rm{BVM}}^{(1)}$ must satisfy~\cite{Durrett16}
\begin{eqnarray}
  \sum_{s_0} Q_{\rm{BVM}}(\rho_R,s_0) \mu_{\rm{BVM}}^{(1)}(s_0) = -1 \qquad \mbox{for all
  $\rho_R=\frac1N,\ldots,\frac{N-1}N$},
  \label{eq:Qg1}\\
  \mu_{\rm{BVM}}^{(1)}(0) = \mu_{\rm{BVM}}^{(1)}(1) = 0.\nonumber
\end{eqnarray}
Inserting~\eref{eq:QincBVM}, \eref{eq:QdecBVM} and \eref{eq:Q0BVM}
into~\eref{eq:Qg1}, we obtain
\begin{equation}
\fl cN\rho_R (1-\rho_R)
\left[\mu_{\rm{BVM}}^{(1)}\left(\rho_R-\frac1N\right) -
  2\mu_{\rm{BVM}}^{(1)}(\rho_R) +
  \mu_{\rm{BVM}}^{(1)}\left(\rho_R+\frac1N\right)\right] = -1.
  \label{eq:g1-1}
\end{equation}
We approximate $\mu_{\rm{BVM}}^{(1)}$ by a smooth function and substitute its Taylor expansion into \eref{eq:g1-1}.
After dropping terms $O(N^{-3})$, which are negligible if $N\gg 1$, we obtain the differential equation
\begin{equation}
\frac {\rmd^2} {\rmd\rho_R^2}
  \mu_{\rm{BVM}}^{(1)}(\rho_R) = -\frac N {c\rho_R(1-\rho_R)}\ .
  \label{eq:g1''}
\end{equation}
The solution is~\cite{Sood_etal08,Redner01}
\begin{equation}
\mu_{\rm{BVM}}^{(1)}(\rho_R) = -\frac Nc\cdot [\rho_R\ln\rho_R +
(1-\rho_R)\ln(1-\rho_R)].
\label{eq:g1BVM}
\end{equation}

We can generalize this approach to higher moments of the consensus
time.
Let $\mu^{(n)}_{\rm{BVM}}(s_0)$ denote the $n$-th moment of the consensus
time conditioned on the initial state $s_0$,
\[
\mu^{(n)}_{\rm{BVM}}(s_0) = E\left[T_{\rm{cons}}^n \mid S(0)=s_0\right].
\]
With the method described in section 1.6.2.2 of Ref.~\cite{Redner01}, one can 
%
%
obtain the following generalization
of~\eref{eq:g1''},
\begin{equation}
\frac {\rmd^2} {\rmd\rho_R^2}
\mu^{(n)}_{\rm{BVM}}(\rho_R) = -\frac{nN} c \frac{\mu^{(n-1)}_{\rm{BVM}}(\rho_R)} {\rho_R(1-\rho_R)}\ .
\label{eq:ode}
\end{equation}
For the second moment, in particular, we find the following result that we have not seen explicitly stated in the previous literature,
\begin{eqnarray}
  \fl \mu^{(2)}_{\rm{BVM}}&(\rho_R) =\label{eq:g2BVM}\\
  \fl &\frac{2N^2}{c^2} \left[ \rho_R
  \ln\rho_R + (1-\rho_R) \ln(1-\rho_R) - \rho_R\Li_2(\rho_R) -
  (1-\rho_R)\Li_2(1-\rho_R) + \frac{\pi^2}6\right],\nonumber
\end{eqnarray}
where $\Li_2$ is the dilogarithm.
%
%
%
%
%

Our objective is to compare the mean $\mu_{\rm{BVM}}^{(1)}(s_0)$ and the standard
deviation 
\begin{equation}
\sigma_{\rm{BVM}}(s_0) = \sqrt{\mu^{(2)}_{\rm{BVM}}(s_0) - \left[{\mu_{\rm{BVM}}^{(1)}}(s_0)\right]^2}
\label{eq:sigmaBVM}
\end{equation}
of the BVM consensus time with those of the CVM.
We will derive that, for $N\gg 1$, the CVM's mean and standard deviation differ by a factor that depends on the rates of copying, externalization and internalization.
We will give the explicit equation~\eref{eq:tauCVM} for this factor below.

\section{An intermediate, two-layered voter model (TLVM)}
\label{sec:tlvm}

\subsection{Motivating and defining the TLVM}
\label{subsec:motivateTLVM}

The CVM is more complicated than the BVM because we need more than one
variable to describe its state.
The agents can have four possible combinations of
external and internal opinions ($Rr$, $Rb$, $Br$ and $Bb$).
Because the sum of agents in these four states is constrained by the
total number $N$ of all agents, the state space in the CVM is three-dimensional.
We can, for example, choose the following set of variables to uniquely
characterize the current state,
\begin{itemize}
  \item $\rho_R$: fraction of agents whose {\it external} opinion is
    red (i.e.\ $Rr$ and $Rb$ agents),
  \item $\rho_r$: fraction of agents whose {\it internal} opinion is
    red (i.e.\ $Rr$ and $Br$ agents),
  \item $\rho_{Rr}$: fraction of $Rr$ agents.
\end{itemize}
In the CVM, the events $\omega_{\rm{ext}}(\alpha, t) = R$ (i.e.\ agent $\alpha$'s {\it external} opinion at time $t$ is red) and $\omega_{\rm{int}}(\alpha, t)=r$ (i.e.\ her {\it internal} opinion is red) generally depend on each other. That is, $P[\omega_{\rm{ext}}(\alpha,t) = R, \omega_{\rm{int}}(\alpha,t) = r] = \rho_{Rr}(t) \neq \rho_R(t) \rho_r(t)$.  
As we will show below, we can calculate the leading-order term of the CVM consensus time for $N\gg1$, but the inequality $\rho_{Rr}(t) \neq \rho_R(t)\rho_r(t)$ complicates the solution. 
For this reason, we will first solve a simplified model where we ignore the dependence between the external and internal layer, i.e.\ we set $\rho_{Rr}(t) = \rho_R(t)\rho_r(t)$. We include the simpler model in this article because it demonstrates, with fewer intermediate steps, the analytic techniques that we will apply later to the CVM.  Although the model does not yet fully capture the CVM dynamics, we will still be able to transfer some of the results directly to the CVM.

As in the CVM, we still distinguish between an external and an internal
layer of opinions and, therefore, call this model the {\it Two-Layered
Voter Model} (TLVM).
We keep two of the CVM's features.
\begin{itemize}
\item Each agent copies the external opinion of a random agent with rate $c$.
\item There is no direct opinion exchange in the internal layer.
\end{itemize}
We can impose the condition $\rho_{Rr}(t) = \rho_R(t)\rho_r(t)$ by introducing links between all $N^2$ pairs formed by one
external and one internal opinion.
\begin{itemize}
\item With rate $\frac eN$, every pair externalizes (i.e.\ the internal
opinion becomes the external opinion).
\item With rate $\frac iN$, each pair internalizes.
\end{itemize}
We divide $e$ and $i$ by $N$ so that the mean time between
externalization and internalization events is equal in the TLVM and CVM.
We list all transitions in the TLVM together with their rates in Table~\ref{tab:TLVMtrans}.

\begin{table}
\caption{\label{tab:TLVMtrans}Transitions from the state $(\rho_R, \rho_r)$ in the TLVM and their rates.}
\begin{tabular*}{\textwidth}{@{}lll}
\br
& How is the new & Transition rate matrix element\\
New state $(x,y)$ & state reached? & $Q_{\rm{TLVM}}[(\rho_R,\rho_r), (x,y)]$\\
\mr
$\left(\rho_R+\frac1N, \rho_r\right)$ & A $B$ agent copies a neighbour & $cN\rho_R(1 - \rho_R) + eN(1 - \rho_R)\rho_r$\\
& with external opinion $R$ or\\
& a $Br$ pair externalizes.\\[3pt]
$\left(\rho_R-\frac1N, \rho_r\right)$ & An $R$ agent copies a neighbour & $cN\rho_R(1 - \rho_R) + eN\rho_R(1-\rho_r)$\\
& with external opinion $B$ or\\
& an $Rb$ pair externalizes.\\[3pt]
$\left(\rho_R, \rho_r+\frac1N\right)$ & An $Rb$ pair internalizes. & $iN\rho_R(1 - \rho_r)$\\[3pt]
$\left(\rho_R, \rho_r-\frac1N\right)$ & A $Br$ pair internalizes. & $iN(1-\rho_R)\rho_r$\\[3pt]
$(\rho_R, \rho_r)$ & Negative sum of all rates & $-2cN\rho_R(1-\rho_R) +$\\
& above. & $\hspace{1.0cm}(e+i)N(2\rho_R\rho_r-\rho_R-\rho_r)$\\
\br
\end{tabular*}
\end{table}

\begin{figure}
\begin{center}
\includegraphics[width=0.5\textwidth]{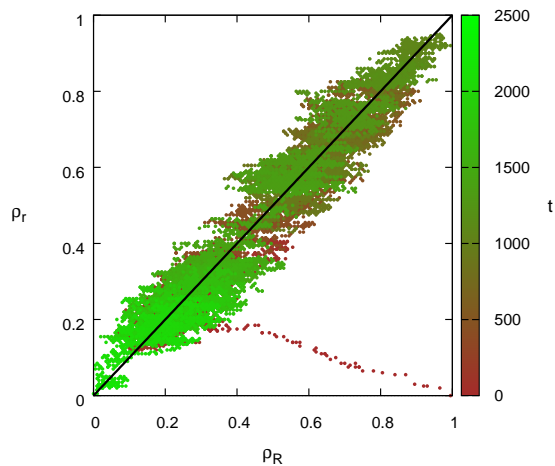}
\caption{\label{fig:TLVMtraj} An example trajectory for the TLVM with $N=200$, $c=1$, $e=0.25$ and $i=0.0625$. We plot the state of the system after every 10th update. In the initial phase (brown dots), the state moves rapidly towards the diagonal line $\rho_R = \rho_r$. Afterwards the system stays near the diagonal until it reaches one of the two consensus states in the lower left and upper right corners of the square. As $N$ gets larger, the state stays closer to the diagonal.
}
\end{center}
\end{figure}

Similar to the BVM, the TLVM possesses a martingale, albeit a
slightly more complex one.
We define the function
\begin{equation}
  m(\rho_R, \rho_r) = \frac{i\rho_R + e\rho_r}{e+i}\ ,
  \label{eq:m}
\end{equation}
which maps the two-dimensional input $(\rho_R, \rho_r)$ onto a real
number in $[0,1]$.
We can view $m$ as the relative proportion of the red opinion that is
present in the combination of the external and internal layer,
prorated by the weights $\frac i {e+i}$ and $\frac e {e+i}$,
respectively.
Let $S_R(t)$ be the random variable that equals the fraction of $R$
agents.\footnote{$S_R$ should not be confused with $\rho_R$. $S_R$ is the function that maps a stochastic configuration of opinions to the fraction $\rho_R$ of red external opinions. Thus, $S_R$ is a random variable whereas $\rho_R$ is a number between $0$ and $1$.}
Similarly, $S_r(t)$ equals the fraction of $r$ agents.
We show in Lemma~\ref{lem:martingale} in the appendix that $m[S_R(t), S_r(t)]$ is a martingale.
In the appendix, we also prove a corollary of this lemma (Corollary~\ref{cor:exit_prob}): if the initial state is $S_R(0)=\rho_R$ and $S_r(0)=\rho_r$, then the probability of reaching a red consensus is $m(\rho_R, \rho_r)$.
As we will see shortly, $m[S_R(t), S_r(t)]$ is a convenient summary statistic to derive the time needed to reach a consensus.

\subsection{The consensus time in the TLVM}
\label{sec:constimeTLVM}

Suppose we are in the state $S_R(t) = \rho_R$, $S_r(t) = \rho_r$ at time $t$.
What can we infer about the state at a later time $t+u$?
In~\ref{app:shortTermEvol}, we prove the following property (Theorem~\ref{thm:shortTermEvol}): after a transient of duration $O[(e+i)^{-1}]$, $S_R(t+u)$ approximately equals $S_r(t+u)$.
If $e$ and $i$ are independent of $N$, the relative error we commit by the approximation $S_R\approx S_r$ is negligible if $N\gg 1$ because the consensus time turns out to be $O(N)$ (see equation~\ref{eq:gTLVM1}).
In other words, we can separate two time scales. On a fast time scale, the expected value of $S_R-S_r$ quickly decays to zero.
On a slow time scale, the system diffuses along the line $\rho_R=\rho_r$ towards one of the absorbing states.
In Figure~\ref{fig:TLVMtraj}, we show a typical run of the TLVM that confirms this conjecture.
We calculated the data shown in Figure~\ref{fig:TLVMtraj} and all other numerical results in this article with the Gillespie algorithm, an exact implementation of the stochastic dynamics~\cite{Gillespie76}.

Because the TLVM spends most of the time in states with $S_R\approx S_r$, we can approximately describe the dynamics in terms of the one-dimensional random variable $M=m[S_R(t), S_r(t)]$, where $m$ is the function defined in~\eref{eq:m}.
In this simplified picture, we use the symbol $\widetilde{\mathbf{Q}}_{\rm{TLVM}}$ for the transition rate matrix.
The nonzero elements of $\widetilde{\mathbf{Q}}_{\rm{TLVM}}$ are
\begin{eqnarray*}
 \fl \widetilde{Q}_{\rm{TLVM}}\left(m, m+\frac i{(e+i)N}\right) = \widetilde{Q}_{\rm{TLVM}}\left(m, m-\frac i{(e+i)N}\right) = (c+e)Nm(1-m),\\
 \fl \widetilde{Q}_{\rm{TLVM}}\left(m, m+\frac e{(e+i)N}\right) = \widetilde{Q}_{\rm{TLVM}}\left(m, m-\frac e{(e+i)N}\right) = iNm(1-m),\\
 \fl \widetilde{Q}_{\rm{TLVM}}(m, m) = -2(c+e+i)Nm(1-m),
\end{eqnarray*}
which we can derive by substituting~\eref{eq:m}
into the transition rates in Table~\ref{tab:TLVMtrans}.

In analogy to the BVM, we define the mean consensus time in the TLVM as
\[
\mu_{\rm{TLVM}}^{(1)}(m_0) = E[T_{\rm{cons}} \mid M(0) = m_0].
\]
It must satisfy the equivalent of~\eref{eq:Qg1},
\[
\sum_{m_0}\widetilde{Q}_{\rm{TLVM}}(m, m_0)\mu_{\rm{TLVM}}^{(1)}(m_0) = -1 
\]
for all $m = \frac{i N_R + e N_r}{N(e+i)}$ with $N_R, N_r\in\{0, 1, \ldots, N\}$ and $m\notin\{0,1\}$.
As in the BVM, we Taylor expand $\mu_{\rm{TLVM}}^{(1)}$ to second order and obtain the differential equation
\[
\frac{\rmd^2}{\rmd m^2} \mu_{\rm{TLVM}}^{(1)}(m) = -\frac{(e+i)^2 N}{i[(c+e)i + e^2]} \cdot \frac 1 {m(1-m)}\ .
\]
The solution, subject to the boundary conditions $\mu_{\rm{TLVM}}^{(1)}(0) = \mu_{\rm{TLVM}}^{(1)}(1) = 0$, is
\begin{equation}
\mu_{\rm{TLVM}}^{(1)}(m) = - \frac{(e+i)^2N}{i[(c+e)i+e^2]} \cdot [m\ln m + (1-m)\ln(1-m)]\ .
\label{eq:gTLVM1}
\end{equation}
The calculation of the higher moments of the consensus time
\[
\mu_{\rm{TLVM}}^{(n)} = E\left[T_{\rm{cons}}^n \mid M(0) = m_0\right]
\]
is similar to~\eref{eq:g2BVM} in the BVM.
%
%
Induction on $n$ shows that
\[
\mu_{\rm{TLVM}}^{(n)}(m) = \left(\frac{c(e+i)^2} {i[(c+e)i + e^2]}\right)^n \mu_{\rm{BVM}}^{(n)}(m).
\]
Because $m \approx \rho_R \approx \rho_r$, the dynamics of the external and internal opinions both behave like BVMs, but with a time that is rescaled by a factor 
\begin{equation}
\tau_{\rm{TLVM}}(c,e,i) = \frac{c(e+i)^2} {i[(c+e)i + e^2]}\ .
\label{eq:tauTLVM}
\end{equation}
In particular, the standard deviations of TLVM and BVM are related by
\begin{equation}
\sigma_{\rm{TLVM}}(m) = \tau_{\rm{TLVM}}(c,e,i)\sigma_{\rm{BVM}}(m).
\label{eq:sigmaTLVM}
\end{equation}
In Figure~\ref{fig:meansdTLVM}, we show, as a numerical confirmation of~\eref{eq:gTLVM1} and~\eref{eq:sigmaTLVM}, the mean and standard deviation of the TLVM consensus time obtained from Monte Carlo simulations.
The numerical results are in excellent agreement with the theoretical predictions (i.e.\ the finite size effect at $N=1000$ does not modify the results substantially).

There are four main conclusions from the TLVM that we will be able to transfer to the CVM. 
(1) Although the state space is two-dimensional, the system spends most of the time near the one-dimensional manifold $\rho_R = \rho_r$. 
(2) We can express the position along this manifold in terms of a function $m$ that has the property that $m(S_R, S_r)$ is a martingale.
(3) After expressing the transition rate matrix in terms of $m$, we can derive the moments of the TLVM consensus time distribution. 
(4) The TLVM moments are related to those of the BVM by $\mu_{\rm{TLVM}}^{(n)} = \tau_{\rm{TLVM}}(c,e,i)^n \mu_{\rm{BVM}}^{(n)}$ for a scale factor $\tau_{\rm{TLVM}}$ that is independent of $n$. We
will encounter similar rules in the CVM.

\begin{figure}
\begin{center}
\includegraphics[width=\textwidth]{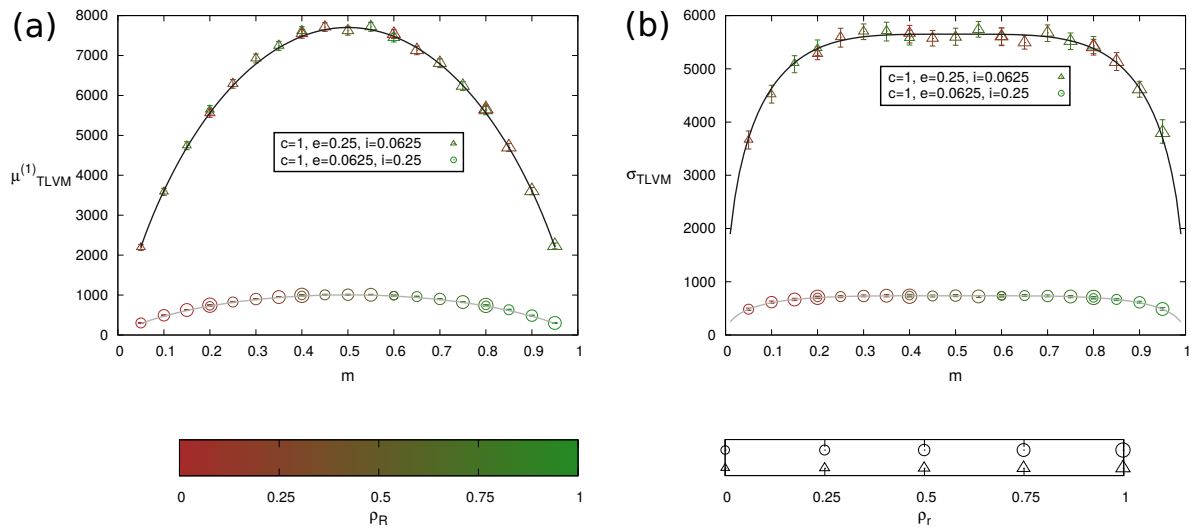}
\caption{The consensus time in the TLVM as a function of the initial value of $m$ (defined in equation~\ref{eq:m}): (a) the mean $\mu_{\rm{TLVM}}^{(1)}$, and (b) the standard deviation $\sigma_{\rm{TLVM}}$. Triangles and circles indicate numerical results from Monte Carlo simulations for two different combinations of $c$, $e$ and $i$. The symbol colour represents the value of  $\rho_R$, the symbol size the value of $\rho_r$. For some values of $m$, we include results from two combinations of $\rho_R$ and $\rho_r$. The symbols are larger than the error bars in several cases. The system size for all simulations is $N=1000$. The curves show the theoretical predictions from equation~\eref{eq:gTLVM1} in panel (a) and from equation~\eref{eq:sigmaTLVM} in panel (b).}
\label{fig:meansdTLVM}
\end{center}
\end{figure}

\section{Solving the Concealed Voter Model (CVM)}
\label{sec:cvm}

\subsection{Short-term evolution of the states in the CVM}
\label{sec:shorttermCVM}

\begin{table}
\caption{\label{tab:CVMtrans}Transitions from the state $(\rho_R, \rho_r, \rho_{Rr})$ in the CVM and their rates.}
\begin{tabular*}{\textwidth}{@{}lll}
\br
& How is the new & Transition rate matrix element\\
New state $(x,y,z)$ & state reached? & $Q_{\rm{CVM}}[(\rho_R,\rho_r,\rho_{Rr}), (x,y,z)]$\\
\mr
$\left(\rho_R+\frac1N, \rho_r, \rho_{Rr} + \frac1N\right)$ & A $Br$ agent externalizes& $cN\rho_R(\rho_r - \rho_{Rr}) + eN(\rho_r - \rho_{Rr})$\\
& or copies a neighbour\\
& with external opinion $R$.\\[3pt]
$\left(\rho_R, \rho_r+\frac1N, \rho_{Rr}+\frac1N\right)$ & An $Rb$ agent internalizes. & $iN(\rho_R - \rho_{Rr})$\\[3pt]
$\left(\rho_R-\frac1N, \rho_r, \rho_{Rr}-\frac1N\right)$ & An $Rr$ agent copies a & $cN\rho_{Rr}(1-\rho_R)$\\
& neighbour with external\\
& opinion $B$.\\[3pt]
$\left(\rho_R+\frac1N, \rho_r, \rho_{Rr}\right)$ & A $Bb$ agent copies a & $cN(1-\rho_R-\rho_r+\rho_{Rr})\rho_R$\\
& neighbour with external\\
& opinion $R$.\\[3pt]
$\left(\rho_R-\frac1N, \rho_r, \rho_{Rr}\right)$ & An $Rb$ agent externalizes & $cN(\rho_R-\rho_{Rr})(1-\rho_R) + $\\
& or copies a neighbour & $\hspace{1cm}eN(\rho_R-\rho_{Rr})$\\
& with external opinion $B$.\\[3pt]
$\left(\rho_R, \rho_r-\frac1N, \rho_{Rr}\right)$ & A $Br$ agent internalizes. & $iN (\rho_r-\rho_{Rr})$\\[3pt]
$(\rho_R, \rho_r, \rho_{Rr})$ & Negative sum of all rates & $-2cN\rho_R(1-\rho_R) +$\\
& above. & $\hspace{1.0cm}(e+i)N(2\rho_{Rr}-\rho_R-\rho_r)$\\
\br
\end{tabular*}
\end{table}

In the CVM, we denote the random variables that map a stochastic configuration at time $t$ to the fractions $\rho_R$, $\rho_r$ and $\rho_{Rr}$ (defined in section~\ref{subsec:motivateTLVM}) by $S_R(t)$, $S_r(t)$ and $S_{Rr}(t)$ respectively.  
Compared to the TLVM, the CVM not only has a three- instead of a two-dimensional state space, but
the transition rate matrix $\mathbf{Q}_{\rm{CVM}}$ also has more
nonzero elements (listed in Table~\ref{tab:CVMtrans}).
Despite the added complexity, some of the TLVM results remain unchanged for the CVM.
In particular, the TLVM martingale $M(t) = m(S_R(t), S_r(t))$, where $m$ is defined in equation~\eref{eq:m}, is also a martingale of the CVM (see Lemma~\ref{eq:TLVMmartApp} in the appendix).
Moreover, as in the TLVM, $m(\rho_R, \rho_r)$ equals the probability of reaching a red consensus if the initial condition satisfies $S_R(0) = \rho_R$ and $S_r(0) = \rho_r$ (see Corollary~\ref{cor:exit_prob}).
Curiously, $m$ does not depend on $\rho_{Rr}$.
If we arrange the four possible combinations of external and internal opinions $Rr$, $Rb$, $Br$ and $Bb$ in a $2\times 2$ contingency table, the marginal frequencies $\rho_R$ and $\rho_r$ fully determine the probability that all agents ultimately agree with the red opinion. The joint distribution of internal and external opinions, which can be derived with the help of $\rho_{Rr}$, does not add more information about the probable consensus opinion.

Besides identical martingales, the TLVM and the CVM also have
equations~\eref{eq:TLVMcondExpectS_R}--\eref{eq:varSr}
in common (see Theorem~\ref{thm:shortTermEvol} in the appendix).
We can deduce from~\eref{eq:TLVMcondExpectS_R}--\eref{eq:varSr} that, after a transient of duration
$O[(e+i)^{-1}]$, the state $(S_R,S_r,S_{Rr})$ of the CVM satisfies
$S_R\approx S_r \approx m(S_R, S_r)$.
If $e$ and $i$ are independent of  $N$ and $N\gg 1$, the transient is negligible compared to the consensus
time so that we can approximate the dynamics of the CVM with a
simplified two-dimensional state space.
Upon setting $\rho_R = \rho_r = m$, the matrix elements in
Table~\ref{tab:CVMtrans} become
\begin{eqnarray}
  \fl &\widetilde{Q}_{\rm{CVM}}\left[(m, \rho_{Rr}), \left(m+\frac i{N(e+i)}, \rho_{Rr} + \frac
    1 N\right)\right] = cNm(m-\rho_{Rr}) + eN(m-\rho_{Rr}),
  \label{eq:QCVM2+i+}\\
  \fl &\widetilde{Q}_{\rm{CVM}}\left[(m, \rho_{Rr}), \left(m + \frac e{N(e+i)}, \rho_{Rr} + \frac
    1 N\right)\right] = iN(m-\rho_{Rr}),
  \label{eq:QCVM2+e+}\\
  \fl &\widetilde{Q}_{\rm{CVM}}\left[(m, \rho_{Rr}), \left(m - \frac i{N(e+i)}, \rho_{Rr} -
    \frac 1 N\right)\right] = cN\rho_{Rr}(1-m),
  \label{eq:QCVM2-i-}\\
  \fl &\widetilde{Q}_{\rm{CVM}}\left[(m, \rho_{Rr}), \left(m + \frac i{N(e+i)},
    \rho_{Rr}\right)\right] = cN(1-2m+\rho_{Rr})m,
  \label{eq:QCVM2+i0}\\
  \fl &\widetilde{Q}_{\rm{CVM}}\left[(m, \rho_{Rr}), \left(m - \frac i{N(e+i)},
    \rho_{Rr}\right)\right] = cN(m-\rho_{Rr})(1-m) + eN(m-\rho_{Rr}),
  \label{eq:QCVM2-i0}\\
  \fl &\widetilde{Q}_{\rm{CVM}}\left[(m, \rho_{Rr}), \left(m - \frac e{N(e+i)}, \rho_{Rr}\right)
    \right] = iN(m - \rho_{Rr}),
  \label{eq:QCVM2-e0}\\
  \fl &\widetilde{Q}_{\rm{CVM}}[(m, \rho_{Rr}), (m, \rho_{Rr})] = -2cNm(1-m) +
    2(e+i)N(\rho_{Rr}-m).
    \label{eq:QCVM200}
\end{eqnarray}
With this two-dimensional approximation of $\mathbf{Q}_{\rm{CVM}}$, we can deduce the short-term evolution of $S_{Rr}$ (see Theorem~\ref{thm:tildeQCVM} in the appendix):
after a transient that lasts no longer than $O[(c+e+i)^{-1}]$, the states in the CVM satisfy $\rho_R\approx \rho_r\approx m$ and $\rho_{Rr}\approx \frac{cm^2+(e+i)m}{c+e+i}$. We show a typical trajectory of a CVM simulation in Figure~\ref{fig:CVMtraj} that confirms this approximation.

\begin{figure}
\begin{center}
\includegraphics[width=0.5\textwidth]{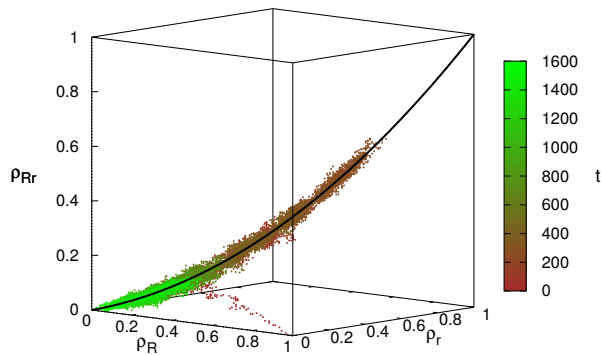}
\caption{\label{fig:CVMtraj} An example trajectory for the CVM with $N=200$, $c=1$, $e=0.25$ and $i=0.0625$. We plot the state of the system after every 10th update.
Similar to the TLVM dynamics in Figure~\ref{fig:TLVMtraj}, the system  moves in the initial phase (brown dots) rapidly towards a one-dimensional curve (black) that connects the two consensus states, which in the CVM are $\rho_R = \rho_r = \rho_{Rr} = 0$ or $\rho_R = \rho_r = \rho_{Rr} = 1$ at opposite ends of the cube. The parametric equations for the curve are $\rho_R = \rho_r = m$ and $\rho_{Rr} = \frac{cm^2 + (e+i)m}{c+e+i}$ for $m\in[0,1]$.
After the initial drift, the system stays in the vicinity of this curve until it reaches one of the consensus states.}
\end{center}
\end{figure}

\subsection{The consensus time in the CVM}

Figure~\ref{fig:CVMtraj} demonstrates that we can neglect the transient (i.e.\ the initial approach towards the black curve) compared to the consensus time.
We can thus approximate the CVM dynamics by a one-dimensional model whose transition rate matrix follows from setting $\rho_{Rr} = \frac{cm^2+(e+i)m}{c+e+i}$ in~\eref{eq:QCVM2+i+}--\eref{eq:QCVM200},
\begin{eqnarray}
\fl \stackrel{\approx}{Q}_{\rm{CVM}}\left(m, m+\frac i{(e+i)N}\right) = \;\stackrel{\approx}{Q}_{\rm{CVM}}\left(m, m-\frac i{(e+i)N}\right)\label{eq:QCVM+i}\\
= c\left(1+\frac e{c+e+i}\right)Nm(1-m),\nonumber\\
\fl \stackrel{\approx}{Q}_{\rm{CVM}}\left(m, m+\frac e{(e+i)N}\right) = \;\stackrel{\approx}{Q}_{\rm{CVM}}\left(m, m-\frac e{(e+i)N}\right) = \frac{ciNm(1-m)}{c+e+i}\ ,\\
\fl \stackrel{\approx}{Q}_{\rm{CVM}}(m,m) = -2c\left(1+\frac{e+i}{c+e+i}\right)Nm(1-m).\label{eq:QCVM0}
\end{eqnarray}
We use a double tilde in the symbol $\doublewidetilde{\mathbf{Q}}_{\rm{CVM}}$ for the transition rate matrix to express that~\eref{eq:QCVM+i}--\eref{eq:QCVM0} are approximations of $\widetilde{\mathbf{Q}}_{\rm{CVM}}$, which in turn is an approximation of the exact CVM transition rate matrix $\mathbf{Q}_{\rm{CVM}}$ given by Table~\ref{tab:CVMtrans}. The relative error that we introduce with these approximations is negligible if $N\gg 1$.

With the one-dimensional approximation $\doublewidetilde{\mathbf{Q}}_{\rm{CVM}}$, the derivation of the mean consensus time is now analogous to that of the TLVM in section~\ref{sec:constimeTLVM}.
We call the mean CVM consensus time 
\[
\mu_{\rm{CVM}}^{(1)}(m_0) = E[T_{\rm{cons}} \mid M(0) = m_0]
\]
so that
\begin{equation}
\sum_{m_0} \stackrel{\approx}{Q}_{\rm{CVM}}(m, m_0) \mu_{\rm{CVM}}^{(1)}(m_0) = -1
\label{eq:QCVMg1}
\end{equation}
if $m\notin\{0,1\}$ and
\begin{equation}
\mu_{\rm{CVM}}^{(1)}(0) = \mu_{\rm{CVM}}^{(1)}(1) = 0.
\label{eq:gCVM1BC}
\end{equation}
Upon inserting~\eref{eq:QCVM+i}--\eref{eq:QCVM0} into~\eref{eq:QCVMg1} and taking the continuum limit, we obtain the differential equation
\[
\frac{\rmd^2}{\rmd m^2}\mu_{\rm{CVM}}^{(1)}(m) = -\frac{(c+e+i)(e+i)^2  N} {ci[ci+(e+i)^2]} \cdot \frac 1{m(1-m)}\ .
\]
The solution, subject to the boundary condition~\eref{eq:gCVM1BC}, is
\begin{equation}
\mu_{\rm{CVM}}^{(1)}(m) = -\frac{(c+e+i)(e+i)^2 N} {ci[(e+i)^2 + ci]}\cdot [m\ln m + (1-m)\ln(1-m)].
\label{eq:gCVM1}
\end{equation}
Comparing~\eref{eq:gCVM1} with the BVM mean consensus time in~\eref{eq:g1BVM}, we find that
\[
\mu_{\rm{CVM}}^{(1)}(m) = \frac{(c+e+i)(e+i)^2}{i[(e+i)^2+ci]}\cdot \mu_{\rm{BVM}}^{(1)}(m),
\]
which generalizes to the higher moments 
as
\begin{equation}
\mu_{\rm{CVM}}^{(n)}(m) = \left(\frac{(c+e+i)(e+i)^2}{i[(e+i)^2+ci]}\right)^n \mu_{\rm{BVM}}^{(n)}(m).
\label{eq:gCVMn}
\end{equation}
It follows that the standard deviation of the consensus time obeys
\begin{equation}
\sigma_{\rm{CVM}}(m) = \frac{(c+e+i)(e+i)^2}{i[(e+i)^2+ci]}\cdot \sigma_{\rm{BVM}}(m),
\label{eq:sigmaCVM}
\end{equation}
where $\sigma_{\rm{BVM}}$ can be calculated from~\eref{eq:g2BVM} and~\eref{eq:sigmaBVM}.
Monte Carlo simulations of the CVM are in excellent agreement with equations~\eref{eq:gCVM1} and~\eref{eq:sigmaCVM}, see Figure~\ref{fig:meansdCVM}.

\begin{figure}
\begin{center}
\includegraphics[width=\textwidth]{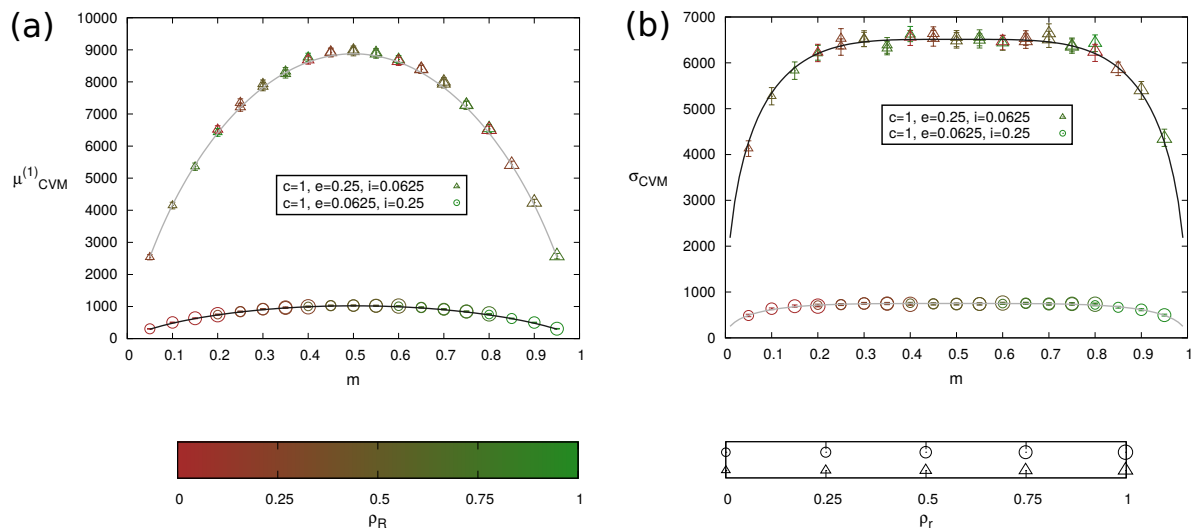}
\caption{(a) The mean consensus time $\mu_{\rm{CVM}}^{(1)}$ of the CVM and (b) the standard deviation $\sigma_{\rm{CVM}}$ as a function of $m$. 
As in Figure~\ref{fig:meansdTLVM}, triangles and circles are Monte Carlo results for two different combinations of $c$, $e$ and $i$.
The symbol colour specifies $\rho_R$, the symbol size $\rho_{r}$.
We also include results for different $\rho_{Rr}$ so that there are multiple measurements for a given value of $m$.
The symbols are larger than the error bars in several cases.
The system size for all simulations is $N=1000$.
The theoretical predictions from equations~\eref{eq:gCVM1} and~\eref{eq:sigmaCVM} are shown as black curves.
}
\label{fig:meansdCVM}
\end{center}
\end{figure}
Equation~\eref{eq:gCVMn} has the remarkable consequence that the CVM consensus time distribution differs from that of the BVM only by the prefactor 
\begin{equation}
\tau_{\rm{CVM}}(c, e, i) = \frac{(c+e+i)(e+i)^2}{i[(e+i)^2+ci]}\ .
\label{eq:tauCVM}
\end{equation}We plot the function $\tau_{\rm{CVM}}$ in Figure~\ref{fig:tauCVM}. 
Interestingly, $\tau_{\rm{CVM}}(c,e,i)>1$ for all $c$, $e$ and $i$ for the following reason.
Because $c$, $e$ and $i$ are rates and hence positive numbers, we must have 
\begin{eqnarray*}
&0 < c < c+e \quad{\rm{and}} \quad 0 < i^2 < (e+i)^2\\
&\Rightarrow 0 < ci^2 < (c+e)(e+i)^2.
\end{eqnarray*}
Next we add $i(e+i)^2$ to the last two terms in the inequality,
\[
0 < i[(e+i)^2 + ci] < (c+e+i)(e+i)^2.
\]
Hence, the denominator in~\eref{eq:tauCVM} is positive and smaller than the numerator, proving $\tau_{\rm{CVM}}(c,e,i) > 1$.
By contrast, the corresponding TLVM scale factor $\tau_{\rm{TLVM}}$, given by equation~\eref{eq:tauTLVM}, can be larger or smaller than $1$.

\begin{figure}
\begin{center}
\includegraphics[width=10cm]{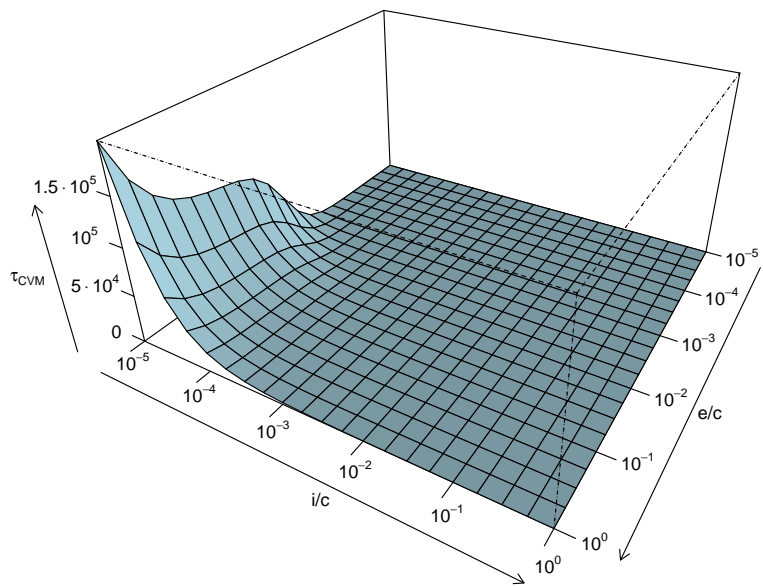}
\caption{The factor $\tau_{\rm{CVM}}$ by which the mean CVM consensus time $\mu_{\rm{CVM}}^{(1)}$ is prolonged compared to the BVM. Because $\tau_{\rm{CVM}}$ only depends on $\frac e c$ and $\frac i c$ (see equation~\ref{eq:tauCVM}), we use these ratios for the coordinate axes. The delay is especially severe when $e$ is large and $i$ is small. In this region, agents tend to be candid about their internal opinions towards the public and they maintain their internal opinions for a long time. In the limit $i\to 0$ and $e\to\infty$, the CVM is similar to the zealot model by Mobilia et al.~\cite{MobiliaEtAl07}, see our discussion in section~\ref{sec:discuss}.}
\label{fig:tauCVM}
\end{center}
\end{figure}

\section{Discussion}
\label{sec:discuss}

Having concealed opinions is ubiquitous on various subjects from politics to personal habits. 
An agent may choose to misrepresent her privately held opinion when it appears to be socially unacceptable.
For example, the agent may disapprove of the latest fashion trend, but still adopt it just to blend in with her acquaintances.
Kuran~\cite{Kuran09} argues that such preference falsification can slow down changes in social norms.
The CVM is a simple model to investigate whether the existence of concealed opinions indeed prolongs the consensus time and, if yes, to what extent.

Our study is related to earlier models on ``partisan'' voters~\cite{MasudaHeteroVot10,MasudaRednerTruth10}. Partisans have an innate and fixed preference for one of the opinions. 
In the limit $i\to0$, the CVM resembles the partisan model: although agents externally accept other opinions, the internal layer never changes. Models with ``zealot'' agents (also called stubborn or inflexible agents~\cite{MobiliaEtAl07,GalamJacobs07,YildizStubborn2013}) are even more restrictive.
Zealots are agents that never change their opinions.
In these models, typically only a small number of agents are zealots, but even as few as two zealots with opposite opinions are enough to prevent a consensus.
In the limiting case $i\to0$ and $e\to\infty$, the CVM corresponds to a zealot model so that the consensus time goes to infinity. 
This explains the high values of $\tau_{\rm{CVM}}(c,e,i)$ in the front left corner of Figure~\ref{fig:tauCVM}.

The partisan model permits the flow of opinions from the internal to the external layer, but only in this direction.
In the CVM, the flow is bidirectional (provided that  $i>0$ and $e>0$) so that agents can hold contrarian opinions for long but not infinite times. The longer an agent has been holding an external opinion, the more likely it is that she has adopted it as an internal one. Thereby, the CVM takes into account Kuran's remark that social norms can become genuinely accepted simply because they have persisted for a long time.

A characteristic feature of the CVM is that, for all externalization and internalization rates, the consensus is on average slower than in the BVM.
Masuda et al.~\cite{MasudaHeteroVot10} have shown that one way to slow down the convergence to consensus in the BVM is the introduction of heterogeneous copying rates. 
In their model, each agent $k$ has a different copying rate $c_k$ chosen from a heavy-tailed probability distribution.
The CVM gives an alternative mechanism for the prolongation of the consensus time.
It keeps the assumption of homogeneity as in the original BVM (i.e.\ neither $c$ nor $e$ nor $i$ depend on the agent). 
The deceleration in the CVM is instead caused by the requirement to reconcile the two layers with each other. The process that leads to reconciliation is fairly complex. Even when the external layer has reached an apparent consensus, the alternative opinion may lurk in the internal layer. This model feature in itself prolongs the consensus time. However, the main reason for the slower consensus is that the lurking opinion can come to the fore by externalization and then spread by copying, returning the system to a state of discord in the external layer.

In spite of this complexity,  we were able to find some simple relations for a complete graph if $N\gg1$. We defined a martingale $M$ that determines the probability that the agents eventually agree on one opinion (``red'') rather than the alternative (``blue'').  By analyzing the stochastic dynamics of $M$, we obtained the leading-order terms for the mean and standard deviation of the consensus time. Notably, the consensus times in the CVM only differ from the BVM by a factor $\tau_{\rm{CVM}}(c,e,i)>1$ that is independent of $N$.
It would be an interesting task for future research to study how different network topologies (e.g.\ Erd\H{o}s-R\'enyi graphs or spatial lattices) or heterogeneity in the rates $c$, $e$, and $i$ change the consensus time in the CVM. 

\ack
We are grateful to K\'aroly Tak\'acs and Zsuzsanna Szvetelszky for exciting discussions about the topic and helpful comments on this manuscript. 
The work was supported by the European Commission (project number FP7-PEOPLE-2012-IEF 6-4564/2013, M.~T.~Gastner), the European Research Council under the European Union’s Horizon 2020 research and innovation programme (grant agreement No 648693, M.~Guly\'as) and NKFIH-OTKA (grant agreements K109215 and K124438, B.~Oborny, and K112929, M.~Guly\'as).

\appendix
\renewcommand{\thetheorem}{\Alph{section}.\arabic{theorem}}
%
%

\section{Derivation of the probability of a red consensus}

We want to calculate the probability that all agents finally agree on the red opinion.
The crucial step is to find a suitable martingale.

\begin{lemma}
The random variable $M(t) = \frac{iS_R+e S_r}{e+i}$ is a martingale of both the TLVM and CVM. 
That is, $M$ satisfies
\begin{equation}
  E[M(t+u) \mid S(t)] = M(t)
  \label{eq:TLVMmartApp}
\end{equation}
for any $u\geq 0$, where $S(t)$ is the state at time $t$,
\begin{equation}
S(t) = 
\cases{
[S_R(t), S_r(t)] & for the TLVM,\\
[S_R(t), S_r(t), S_{Rr}(t)] & for the CVM.
}
\label{eq:S}
\end{equation}
\label{lem:martingale}
\end{lemma}

\begin{proof}
Let $\mathbf{Q}$ be the transition rate matrix $\mathbf{Q}_{\rm{TLVM}}$ or $\mathbf{Q}_{\rm{CVM}}$, respectively.
Because $S$ is a Markov chain with transition rate matrix
$\mathbf{Q}$, we must have
\begin{equation}
P[S(t+u) = x \mid S(t) = s] = [\exp(u\mathbf{Q})]_{s, x},
\label{eq:P_from_expApp}
\end{equation}
where the right-hand side is the $(s, x)$-th element in the matrix
$\exp(u\mathbf{Q})$.
Hence,
\begin{equation}
  E[M(t+u) - M(t) \mid S(t) = s] =
  \sum_x\left[\exp(u\mathbf{Q})-\mathbf{1}\right]_{s,x} m(x),
  \label{eq:M|SApp}
\end{equation}
where $\mathbf{1}$ is the identity matrix and $m$ is defined by~\eref{eq:m}.
Next we take the time derivative of~\eref{eq:M|SApp},
\begin{eqnarray*}
  \fl\frac{\rmd}{\rmd u} E[M(t+u) - M(t) \mid S(t) = s]
  &=
    \sum_x[\exp(u\mathbf{Q})\mathbf{Q}]_{s,x}m(x)\\
  &=
    \sum_y[\exp(u\mathbf{Q})]_{s,y}\sum_xQ(y,x)
    m(x).
\end{eqnarray*}
If we can show
\begin{equation}
  \sum_x Q(y,x)m(x) = 0
  \label{eq:Qm=0App}
\end{equation}
for all states $y = (\rho_R, \rho_r)$ in the TLVM and $y = (\rho_R, \rho_r, \rho_{Rr})$ in the CVM, then $E[M(t+u) - M(t) \mid S(t)
= s]$ is independent of $u$ so that~\eref{eq:TLVMmartApp} must be true.
We can verify~\eref{eq:Qm=0App} by substituting the elements of $\mathbf{Q}_{\rm{TLVM}}$ or $\mathbf{Q}_{\rm{CVM}}$ from Table~\ref{tab:TLVMtrans} or~\ref{tab:CVMtrans}, respectively, into~\eref{eq:Qm=0App}.
\end{proof}

The probability of a red consensus now follows from the following corollary.

\begin{corollary}
If the initial state is 
\[
S(0) = 
\cases{
(\rho_R, \rho_r) & in the TLVM,\\
(\rho_R, \rho_r, \rho_{Rr}) & in the CVM,
}
\]
then the probability of reaching a red consensus is $m(\rho_R, \rho_r)$, where $m$ is defined by~\eref{eq:m}.
\label{cor:exit_prob}
\end{corollary}

\begin{proof}
The time $T_{\rm{cons}}$ until the consensus is reached is a stopping
time and $M$ is a martingale with time-independent upper and lower
bounds, namely $0$ and $1$.
Thus, $E[M(T_{\rm{cons}})] = E[M(0)]$.
Because $M(T_{\rm{cons}})$ can only be $0$ (blue consensus) or $1$
(red consensus), we must have
$P({\rm{red\ consensus}}) = EM(T_{\rm{cons}}) = EM(0) = m(\rho_R,
\rho_r)$.
\end{proof}

\section{Derivation of the short-term evolution of the TLVM and CVM}
\label{app:shortTermEvol}

In this appendix, we derive the conditional mean and variance of a future state shortly after the system was in state $S$ defined by~\eref{eq:S}.

\begin{theorem}
We define $s = (\rho_R, \rho_r)$ in the TLVM and $s = (\rho_R, \rho_r, \rho_{Rr})$ in the CVM.
In both the TLVM and the CVM, the conditional expectation of a future state after a time $u\geq 0$ is given by
\begin{eqnarray}
\fl E\left[S_R(t+u) \mid S(t) = s\right] =
\frac{i\rho_R + e\rho_r + e[\rho_R-\rho_r]\exp[-(e+i)u]} {e+i}\ ,
\label{eq:TLVMcondExpectS_R}\\
\fl E\left[S_r(t+u) \mid S(t) = s\right] =
\frac{i\rho_R + e\rho_r + i[\rho_r-\rho_R]\exp[-(e+i)u]} {e+i}\ ,
\label{eq:TLVMcondExpectS_r}
\end{eqnarray}
independent of $N$, whereas the conditional variances satisfy
\begin{eqnarray}
\var\left[S_R(t+u) \mid S(t) = s\right] = O(N^{-1}),
\label{eq:varSR}\\
\var\left[S_r(t+u) \mid S(t) = s\right] = O(N^{-1}).
\label{eq:varSr}
\end{eqnarray}
Therefore, if $N\gg 1$, the exponential decay in~\eref{eq:TLVMcondExpectS_R} and~\eref{eq:TLVMcondExpectS_r} is much faster than the increase in the variances.
\label{thm:shortTermEvol}
\end{theorem}

\begin{proof}
We outline the proof of~\eref{eq:TLVMcondExpectS_R}.
Equations~\eref{eq:TLVMcondExpectS_r}--\eref{eq:varSr} can be derived similarly.

From~\eref{eq:P_from_expApp}, it follows that the conditional expectation is
\begin{eqnarray}
\fl E[S_R(t+u) \mid S(t) = s]
  &= \sum_x x_R [\exp(u\mathbf{Q})]_{s,x}
    \nonumber\\
  &= \sum_{k=0}^\infty \left\{\frac{u^k}{k!} \sum_x x_R Q^k(s, x)\right\},
  \label{eq:ES_R|St}
\end{eqnarray}
where $\mathbf Q$ is the transition rate matrix $\mathbf Q_{\rm{TLVM}}$ or $\mathbf Q_{\rm{CVM}}$ from Table~\ref{tab:TLVMtrans} or~\ref{tab:CVMtrans}, respectively.
The summation in~\eref{eq:ES_R|St} is over all states $x=(x_R, x_r)$ in the TLVM or $x=(x_R, x_r, x_{Rr})$ in the CVM.
Induction on $k$ proves that
\begin{equation}
\sum_x x_R Q^k(s, x) = (-1)^{k-1} e(e+i)^{k-1}(\rho_r - \rho_R),
\label{eq:sum_xR_Qk}
\end{equation}
for $k=1,2,\ldots$
Substituting~\eref{eq:sum_xR_Qk} into~\eref{eq:ES_R|St}, we obtain
\begin{eqnarray*}
\fl E[S_R(t+u) \mid S(t) = s]
&= \rho_R + \frac{e(\rho_R-\rho_r)}{e+i} \sum_{k=1}^\infty \frac{u^k}{k!}(-1)^k (e+i)^k\\
&=\frac{i\rho_R + e\rho_r + e(\rho_R-\rho_r)\exp[-(e+i)u]} {e+i}\ .
\end{eqnarray*}
\end{proof}

Our next result characterizes the short-term evolution of $S_{Rr}$ in the two-dimensional approximation of the CVM defined by the transition rate matrix $\widetilde{\mathbf{Q}}_{\rm{CVM}}$ given by equations~\eref{eq:QCVM2+i+}--\eref{eq:QCVM200}. Because in this approximation $S_R = S_r$, it is sufficient to characterize the state $S = (S_R, S_r, S_{Rr})$ by two values: the value of $m$ defined in Eq.~\ref{eq:m} and the value of $\rho_{Rr}$ assumed by the random variable $S_{Rr}$.

\begin{theorem}
The conditional expectation of $S_{Rr}$ satisfies
\begin{eqnarray}
  \fl E[S_{Rr}(t+u) \mid S(t) = (m, \rho_{Rr})]
  \label{eq:EcondSRr}\\
  \fl = \frac{cm^2 + (e+i)m}{c+e+i} + \left(\rho_{Rr} - \frac{cm^2 +
  (e+i)m}{c+e+i}\right) \exp[-(c+e+i)u] + O(N^{-1}).
  \nonumber
\end{eqnarray}
Thus, in the limit of large $N$, the conditional expectation approaches $\frac{cm^2 + (e+i)m}{c+e+i}$ after a transient of
duration $O[(c+e+i)^{-1}]$. The conditional variance satisfies 
\[
\var[S_{Rr}(t+u) \mid S(t) = (m,
\rho_{Rr})] = O(N^{-1})
\] 
so that, during the transient, the increase in the conditional variance is negligible if $N\gg 1$.
\label{thm:tildeQCVM}
\end{theorem}

\begin{proof}
One can show by induction on $k$ that
\begin{eqnarray*}
    \fl \sum_{(x_m, x_{Rr})} &x_m^a x_{Rr}(\widetilde{Q}_{\rm{CVM}})^k[(m, \rho_{Rr}), (x_m,
      x_{Rr})] =\\
    \fl & (-1)^{k-1} (c+e+i)^{k-1}m^a [cm^2 + (e+i)m - (c+e+i)\rho_{Rr}] +
      O(N^{-1})
\end{eqnarray*}
for all nonnegative integers $a$ and $k=1,2,\ldots$
For the special case $a=0$, it follows that
\begin{eqnarray}
    \fl &E[S_{Rr}(t+u) \mid S(t) = (m, \rho_{Rr})] = \sum_{k=0}^\infty \left\{\frac{u^k}{k!}\sum_{(x_m, x_{Rr})}
      x_{Rr} (\widetilde{Q}_{\rm{CVM}})^k[(m, \rho_{Rr}), (x_m,x_{Rr})]\right\}\nonumber\\
    \fl &= \rho_{Rr} - \frac{cm^2 + (e+i)m -
      (c+e+i)\rho_{Rr}}{c+e+i}\sum_{k=1}^\infty\frac{[-(c+e+i)u]^k}{k!}
      + O(N^{-1})\nonumber\\
    \fl &= \rho_{Rr} + \left(\rho_{Rr} -
      \frac{cm^2+(e+i)m}{c+e+i}\right)[\exp(-(c+e+i)u) -
      1] + O(N^{-1}).\label{eq:ESRrProof}
\end{eqnarray}
Rearranging the terms in~\eref{eq:ESRrProof} proves~\eref{eq:EcondSRr}. 
A similar argument proves that the conditional variance is $O(N^{-1})$.
\end{proof}


\section*{References}

\begin{thebibliography}{10}

\bibitem{CliffordSudbury73}
P.~Clifford and A.~Sudbury.
\newblock A model for spatial conflict.
\newblock {\em Biometrika}, 60(3):581--588, 1973.

\bibitem{HolleyLiggett75}
R.~A. Holley and T.~M. Liggett.
\newblock Ergodic theorems for weakly interacting infinite systems and the
  voter model.
\newblock {\em Ann. Probab.}, 3(4):643--663, 1975.

\bibitem{Liggett99}
T.~M. Liggett.
\newblock {\em Stochastic interacting systems: contact, voter, and exclusion
  processes}.
\newblock Springer, Berlin, 1999.

\bibitem{CastellanoEtAl09}
C.~Castellano, S.~Fortunato, and V.~Loreto.
\newblock Statistical physics of social dynamics.
\newblock {\em Rev. Mod. Phys.}, 81(2):591--646, 2009.

\bibitem{FernandezGracia_etal14}
J.~Fern\'andez-Gracia, K.~Suchecki, J.~J. Ramasco, M.~San~Miguel, and V.~M.
  Egu\'{\i}luz.
\newblock Is the voter model a model for voters?
\newblock {\em Phys. Rev. Lett.}, 112(15):158701, 2014.

\bibitem{FrachebourgKrapivsky96}
L.~Frachebourg and P.~L. Krapivsky.
\newblock Exact results for kinetics of catalytic reactions.
\newblock {\em Phys. Rev. E}, 53(4):R3009--R3012, 1996.

\bibitem{DurrettLevin96}
R.~Durrett and S.~Levin.
\newblock Spatial models for species-area curves.
\newblock {\em J. Theor. Biol.}, 179(2):119--127, 1996.

\bibitem{ChaveLeigh02}
J.~Chave and E.~G. Leigh.
\newblock A spatially explicit neutral model of $\beta$-diversity in tropical
  forests.
\newblock {\em Theor. Popul. Biol.}, 62(2):153--168, 2002.

\bibitem{BorileEtAl14}
C.~Borile, P.~D. Pra, M.~Fischer, M.~Formentin, and A.~Maritan.
\newblock Time to absorption for a heterogeneous neutral competition model.
\newblock {\em J. Stat. Phys.}, 156(1):119--130, 2014.

\bibitem{RavaszEtAl04}
M.~Ravasz, G.~Szab\'o, and A.~Szolnoki.
\newblock Spreading of families in cyclic predator-prey models.
\newblock {\em Phys. Rev. E}, 70(1):012901, 2004.

\bibitem{VazquezConstrained04}
F.~Vazquez and S.~Redner.
\newblock Ultimate fate of constrained voters.
\newblock {\em J. Phys. A}, 37(35):8479, 2004.

\bibitem{LambiotteVacillating07}
R.~Lambiotte and S.~Redner.
\newblock Dynamics of vacillating voters.
\newblock {\em J. Stat. Mech. Theor. Exp.}, 2007:L10001, 2007.

\bibitem{Sood_etal08}
V.~Sood, T.~Antal, and S.~Redner.
\newblock Voter models on heterogeneous networks.
\newblock {\em Phys. Rev. E}, 77(4):041121, 2008.

\bibitem{MasudaOpControl15}
N.~Masuda.
\newblock Opinion control in complex networks.
\newblock {\em New J. Phys.}, 17(3):033031, 2015.

\bibitem{Serrano_etal09}
M.~\'Angeles~Serrano, K.~Klemm, F.~Vazquez, V.~M. Egu\`iluz, and M.~San~Miguel.
\newblock Conservation laws for voter-like models on random directed networks.
\newblock {\em J. Stat. Mech. Theor. Exp.}, 2009(10):P10024, 2009.

\bibitem{SoodRedner05}
V.~Sood and S.~Redner.
\newblock Voter model on heterogeneous graphs.
\newblock {\em Phys. Rev. Lett.}, 94(17):178701, 2005.

\bibitem{Kuran09}
T.~Kuran.
\newblock {\em Private truths, public lies}.
\newblock Harvard Univ. Press, Cambridge, 1995.

\bibitem{MasudaHeteroVot10}
N.~Masuda, N.~Gibert, and S.~Redner.
\newblock Heterogeneous voter models.
\newblock {\em Phys. Rev. E}, 82(1):010103(R), 2010.

\bibitem{MasudaRednerTruth10}
N.~Masuda and S.~Redner.
\newblock Can partisan voting lead to truth?
\newblock {\em J. Stat. Mech. Theor. Exp.}, 2011:L02002, 2011.

\bibitem{Gastner15}
M.~T. Gastner.
\newblock The {I}sing chain constrained to an even or odd number of positive
  spins.
\newblock {\em J. Stat. Mech. Theor. Exp.}, 2015(3):P03004, 2015.

\bibitem{Redner01}
S.~Redner.
\newblock {\em A Guide to First-Passage Processes}.
\newblock Cambridge University Press, Cambridge, 2001.

\bibitem{Suchecki_etal05}
K.~Suchecki, V.~M. Egu\'iluz, and M.~San Miguel.
\newblock Conservation laws for the voter model in complex networks.
\newblock {\em Europhys.~Lett.}, 69(2):228, 2005.

\bibitem{Durrett16}
R.~Durrett.
\newblock {\em Essentials of stochastic processes}.
\newblock Springer, Cham, 3rd edition, 2016.

\bibitem{Gillespie76}
D.~T. Gillespie.
\newblock A general method for numerically simulating the stochastic time
  evolution of coupled chemical reactions.
\newblock {\em J. Comput. Phys.}, 22(4):403--434, 1976.

\bibitem{MobiliaEtAl07}
M.~Mobilia, A.~Petersen, and S.~Redner.
\newblock On the role of zealotry in the voter model.
\newblock {\em J. Stat. Mech. Theor. Exp.}, 2007(08):P08029, 2007.

\bibitem{GalamJacobs07}
S.~Galam and F.~Jacobs.
\newblock The role of inflexible minorities in the breaking of democratic
  opinion dynamics.
\newblock {\em Physica A}, 381:366 -- 376, 2007.

\bibitem{YildizStubborn2013}
E.~Yildiz, A.~Ozdaglar, D.~Acemoglu, A.~Saberi, and A.~Scaglione.
\newblock Binary opinion dynamics with stubborn agents.
\newblock {\em ACM Trans. Econ. Comput.}, 1(4):19:1--19:30, 2013.

\end{thebibliography}
\end{document}